\documentclass[11pt]{article}

\usepackage[utf8]{inputenc}
\usepackage[a4paper,margin=1in]{geometry}
\usepackage{amsmath,amssymb}
\usepackage{graphicx}
\usepackage{booktabs}
\usepackage{tabularx}
\usepackage{array}
\usepackage[section]{placeins}
\usepackage{microtype}
\usepackage[round,authoryear]{natbib}
\usepackage{xcolor}
\usepackage{fancyvrb}
\usepackage{tikz}
\usetikzlibrary{arrows.meta,positioning}
\usepackage[font=small,labelfont=bf]{caption}
\usepackage[colorlinks=true,linkcolor=blue!50!black,citecolor=blue!50!black,urlcolor=blue!50!black]{hyperref}

\newcommand{\pkg}[1]{\textbf{#1}}
\newcommand{\code}[1]{\texttt{#1}}
\DefineVerbatimEnvironment{Rcode}{Verbatim}{fontsize=\small,xleftmargin=1.5em}
\newtheorem{proposition}{Proposition}
\newenvironment{proof}{\par\noindent\textit{Proof.}\ }{\hfill$\square$\par\medskip}

\title{When Scatter Plots Hide the Data:\\Diagnosing Ties and Choosing the Right Display in \textsf{R}}
\author{Sameera Viswakula\\
  Department of Statistics, University of Colombo, Sri Lanka\\
  \href{mailto:sam@stat.cmb.ac.lk}{\texttt{sam@stat.cmb.ac.lk}}}
\date{}

\begin{document}
\maketitle

\begin{abstract}
Scatter plots draw one mark per observation, so observations that share a location are hidden. In a two-item Likert survey of 2,757 respondents, 98.7\% of the data lie under other points and the plot shows only 36 dots. The density-distribution sunflower plot of Dupont and Plummer addresses this by drawing raw points in sparse regions and sunflowers, whose petals count observations, in dense ones. It is not available in \textsf{R}, and its single petal weight for dense cells fails when cell counts span more than an order of magnitude: small dense cells collapse to one petal, and the count a reader decodes can be wrong by 150\% or more. We propose geometric petal tiers, in which each cell uses the smallest weight among $1, b, b^2, \ldots$ observations per petal that keeps its flower within $P$ petals. Every flower then has between about $P/b$ and $P$ petals, and the relative error of every decoded count is below $b/(2P)$, 17.9\% at the defaults $b = 5$ and $P = 14$. We also define tie diagnostics and decision rules that use them to choose between an ordinary scatter plot, jitter or transparency, a classic sunflower plot and geometric tiers. Both are implemented in the \pkg{ggplot2} extension \pkg{ggsunflower}. We illustrate the methods on survey, automotive, historical and large commercial data, and in a simulation over 35 conditions geometric tiers keep the worst-cell error at or below 14\% where a single weight reaches a median of 251\%.
\end{abstract}

\noindent\textbf{Software:} \url{https://github.com/sameeranc/ggsunflower}\\[2pt]

\section{Introduction}\label{sec:intro}

The scatter plot is the default display for two numeric variables, and it rests on an assumption that is rarely checked: that each observation gets its own visible mark. When observations share a location, the marks are drawn on top of each other and the plot shows the \emph{support} of the data rather than its \emph{distribution}. Ties are common. Likert items, counts, ages in years, rounded measurements and large data sets recorded on a finite grid all produce them.

Figure~\ref{fig:bfi} shows the problem at its most extreme. Two agreeableness items from the Big Five inventory (\code{bfi} in \pkg{psych}; \citealp{psych}) are answered on a 1--6 scale by 2,757 respondents. The ordinary scatter plot (panel~a) shows a complete $6 \times 6$ grid of 36 identical dots: 98.7\% of the respondents are invisible, and the plot cannot distinguish a cell holding 3 respondents from one holding 483. \citet{galton1886} met the same problem with heights recorded to the nearest inch and published his data as a table of counts rather than a scatter plot.

\begin{figure}[!htbp]
  \centering
  \includegraphics[width=\textwidth]{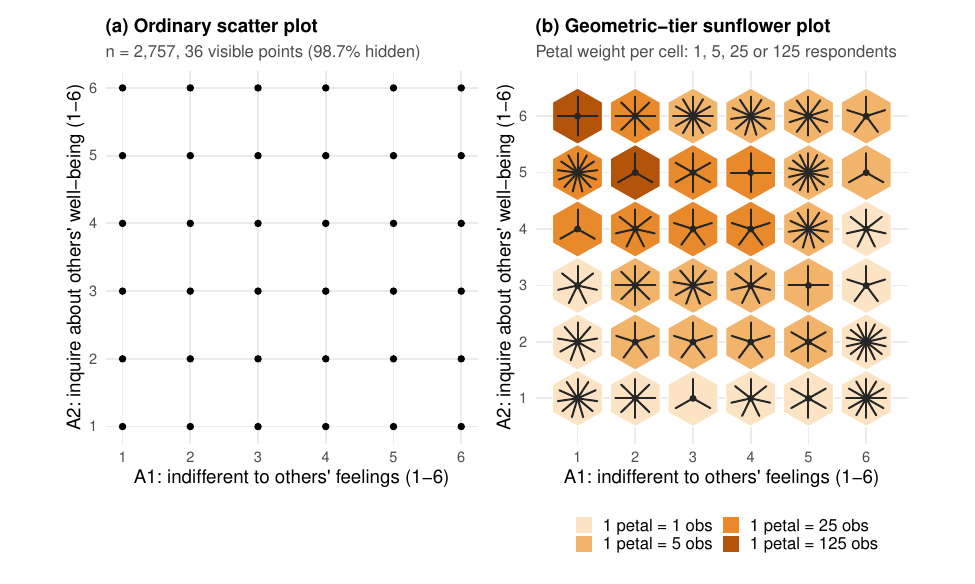}
  \caption{Two Likert items from \code{psych::bfi}. (a) The ordinary scatter plot shows 36 dots for 2,757 respondents. (b) The geometric-tier sunflower plot drawn by \code{geom\_sunflower()} shows where the respondents are, and each count can be read as petals $\times$ weight.}
  \label{fig:bfi}
\end{figure}

Many remedies exist, and each works in some regimes and fails in others (Table~\ref{tab:remedies}). Jitter separates a few coincident points but distorts discrete structure and still overlaps when repeats are many. Transparency accumulates darkness but saturates and cannot be read as a count. Point area (\code{geom\_count()} in \pkg{ggplot2}; \citealp{wickham2016}) is hard to compare precisely \citep{cleveland1984perception}. Hexagonal binning \citep{carr1987} scales to large data but absorbs sparse tails into bins and encodes counts by colour, which readers decode imprecisely.

\begin{table}[!htbp]
  \centering\small
  \caption{Common remedies for repeated observations.}
  \label{tab:remedies}
  \begin{tabularx}{\textwidth}{@{}>{\raggedright}p{3.1cm} >{\raggedright}p{2.9cm} >{\raggedright}X >{\raggedright\arraybackslash}X@{}}
    \toprule
    Remedy & Encoding of repeats & Works when & Fails when \\
    \midrule
    Jitter & random displacement & few repeats per location & discrete structure is distorted; many repeats still overlap \\
    Transparency & accumulated darkness & moderate $n$, continuous data & darkness saturates; counts cannot be read \\
    \code{geom\_count()} & point area & few distinct locations & areas are hard to compare; no exact counts \\
    Hexagonal binning & colour of a cell & large $n$, continuous data & sparse tails vanish into bins; colour is read imprecisely \\
    Sunflower plot & one petal per observation & at most about 14 repeats & petals merge into blobs \\
    Density-distribution sunflower & raw points, light and dark flowers & counts within about one order of magnitude & small dense cells collapse to one petal (Section~\ref{sec:fail}) \\
    \bottomrule
  \end{tabularx}
\end{table}

The sunflower plot \citep{cleveland1984faces} encodes counts by petals, which readers can count, and the density-distribution sunflower plot of \citet{dupont2003,dupont2005} extends it to dense data by giving the petals of dense cells a weight of several observations. This family is the only one in Table~\ref{tab:remedies} that lets a reader recover counts, but it has two gaps. It is available in Stata but not in \textsf{R} \citep{rcore}: base \textsf{R}'s \code{sunflowerplot()} draws one petal per observation, and no \pkg{ggplot2} extension draws petals. And its single petal weight breaks down when counts span more than an order of magnitude, which is exactly when ties matter most.

This paper makes four contributions:
\begin{enumerate}
  \item \textbf{Geometric petal tiers} (Section~\ref{sec:tiers}), a petal rule under which every flower has between about $P/b$ and $P$ petals and every decoded count is within $b/(2P)$ of the truth.
  \item \textbf{Tie diagnostics} (Section~\ref{sec:select}), computed from one frequency table, that measure how much a scatter plot hides.
  \item \textbf{Decision rules} (Section~\ref{sec:select}) that use the diagnostics to choose a display automatically and explain the choice.
  \item \textbf{The \textsf{R} package \pkg{ggsunflower}} (Section~\ref{sec:package}), which implements both as \pkg{ggplot2} layers and accepts pre-aggregated counts from databases and Spark.
\end{enumerate}

\section{Sunflower plots and where they fail}\label{sec:fail}

A sunflower plot replaces the $k$ observations at one location by a dot with $k$ petals, short line segments at equal angles. Petals are counted rather than estimated, which makes them more precise than area or colour, but only while they can be told apart: beyond about 14 petals a flower becomes a solid disc.

\paragraph{The density-distribution rule.} \citet{dupont2003} partition the plane into hexagonal cells and draw each cell in one of three ways according to its count $k$: raw points when $k < L$; a \emph{light} sunflower with $k$ petals, one per observation, when $L \le k \le U$; and a \emph{dark} sunflower with $\operatorname{round}(k/w)$ petals, each worth $w$ observations, when $k > U$. We use $L = 3$ and $U = 13$ throughout. The weight $w$ is a single number for the whole plot. To keep the densest cell readable it must be at least $m/P$, where $m$ is the largest cell count and $P$ the maximum number of petals; we use $w = \lceil m/P \rceil$ with $P = 14$.

\paragraph{Why one weight fails.} The rule works while dark counts stay within about one order of magnitude, so that every dark cell gets a reasonable number of petals. When $m$ is large, $w$ is large, and a cell just above $U$ gets $\operatorname{round}(k/w)$ petals, often zero rounded up to one. In \code{bfi}, $m = 483$ gives $w = 35$. Every dark cell with 14 to 52 respondents is then drawn with a single petal and decoded as 35, and a cell of 14 is overstated by 150\%. Twelve of the 26 dark cells in \code{bfi} have one petal, so much of the dark region carries almost no information (Figure~\ref{fig:bfi4}, panel~c). The fault is not the tiers but the single weight: no one $w$ can serve cells of 14 and of 483.

\begin{figure}[!htbp]
  \centering
  \includegraphics[width=\textwidth]{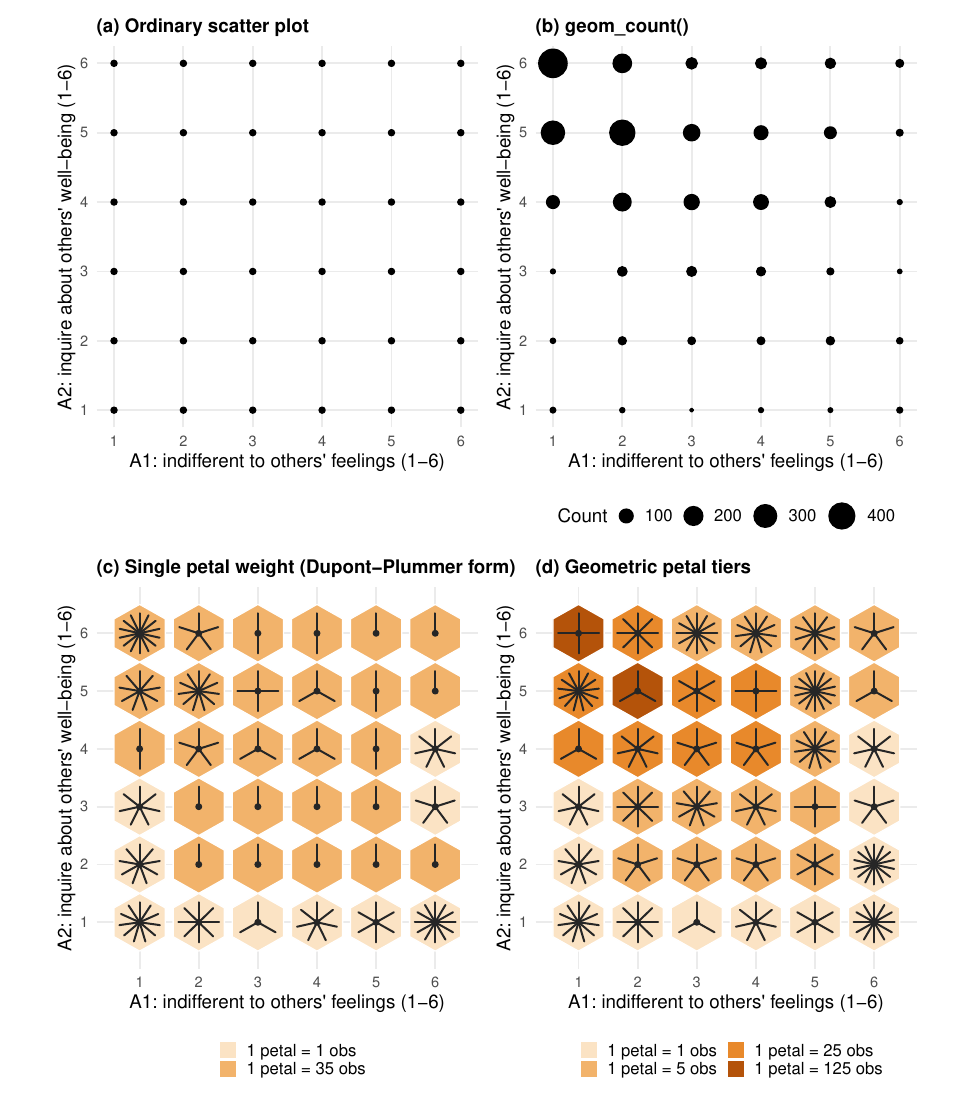}
  \caption{\code{psych::bfi} drawn four ways. (a) Ordinary scatter plot. (b) \code{geom\_count()}. (c) Density-distribution sunflower plot with one dark weight ($w = 35$): 12 of the 26 dark cells collapse to one petal. (d) Geometric petal tiers (Section~\ref{sec:tiers}).}
  \label{fig:bfi4}
\end{figure}

The problem grows with the data. In the simulation of Section~\ref{sec:sim}, the worst-cell error of the single-weight rule rises from 26\% at $n = 1{,}000$ to several thousand per cent at $n = 100{,}000$, because $m$ grows with $n$ while the smallest dark count stays at $U + 1$.

\section{Geometric petal tiers}\label{sec:tiers}

Instead of one weight for every dense cell, each cell chooses its own weight from a geometric ladder, the smallest that keeps its flower within $P$ petals.

\paragraph{Cells.} For discrete data each distinct location is a cell (\emph{exact mode}). Otherwise observations are assigned to a hexagonal lattice, built as two offset rectangular lattices in standardised coordinates, each point going to the nearer centre (\emph{binned mode}). Cells with fewer than $L$ observations show their raw points, as in \citet{dupont2003}.

\paragraph{The rule.} Fix a base $b \ge 2$ and a petal limit $P$. A cell with $k \ge L$ observations is placed in tier
\begin{equation}\label{eq:tier}
  t(k) = \max\bigl\{0,\ \lceil \log_b (k/P) \rceil\bigr\},
\end{equation}
and drawn with weight $w(k) = b^{t(k)}$ observations per petal and $p(k) = \operatorname{round}\bigl(k / w(k)\bigr)$ petals. Tier 0 is the classic sunflower, one petal per observation, used whenever $k \le P$. The hexagon behind the flower is filled with a sequential colour that darkens with the tier, and the legend states the weight of each tier present (``1 petal = 25 obs'').

\begin{proposition}\label{prop:bounds}
Let $\hat k = p(k)\, w(k)$ be the count a reader decodes. For every $k \ge L$:
\begin{enumerate}
  \item[(a)] $p(k) \le P$, and $p(k) \ge \operatorname{round}(P/b)$ whenever $t(k) \ge 1$;
  \item[(b)] $\hat k = k$ if $t(k) = 0$, and $|\hat k - k| / k < b/(2P)$ otherwise;
  \item[(c)] $t(k)$ is non-decreasing in $k$, and $p(k)$ is non-decreasing in $k$ within a tier;
  \item[(d)] if $m \le P$ only tier 0 occurs, and if $P < m \le Pb$ only tiers 0 and 1 occur.
\end{enumerate}
\end{proposition}

\begin{proof}
For $t = t(k) \ge 1$, the definition gives $b^{t-1} < k/P \le b^{t}$, hence $P/b < k/w \le P$ with $w = b^t$. Rounding is monotone, so $\operatorname{round}(P/b) \le p(k) \le P$, which is (a). Since $|\hat k - k| = w\,|\operatorname{round}(k/w) - k/w| \le w/2$ and $k > Pw/b$, the relative error is below $(w/2)/(Pw/b) = b/(2P)$, which is (b). Both $\lceil \log_b(k/P)\rceil$ and $\operatorname{round}(k/w)$ for fixed $w$ are non-decreasing in $k$, which is (c), and (d) follows from \eqref{eq:tier}.
\end{proof}

With $b = 5$ and $P = 14$, every flower above tier 0 has 3 to 14 petals, so none collapses to one, and every decoded count is within 17.9\% of the truth. The single-weight rule has no such bound: its error for a cell just above $U$ grows in proportion to $m$. Part (c) means that a darker hexagon always holds more observations than a lighter one, and within one shade more petals mean more observations. Part (d) means that geometric tiers \emph{generalise} the existing rules: they reduce to the classic sunflower plot when $m \le P$ and to a two-tier density-distribution plot with dark weight $b$ when $m \le Pb$, and differ from \citet{dupont2003} only when $m > Pb$, which is where the single weight fails.

\paragraph{Choosing $b$ and $P$.} $P = 14$ sits just above the light-flower limit $U = 13$, so tier 0 covers the same counts as a light flower. The base trades precision against the number of tiers: the error bound $b/(2P)$ favours a small $b$, while the number of shades a reader must tell apart, $\lceil \log_b(m/P) \rceil + 1$, favours a large one. We use $b = 5$, which gives $1, 5, 25, 125, 625, \ldots$ observations per petal, easy multiples to decode, and at most four tiers up to $m = 1{,}750$. For very large cells $b = 10$ keeps the number of tiers small at the cost of a 36\% bound.

\begin{figure}[!htbp]
  \centering
  \includegraphics[width=\textwidth]{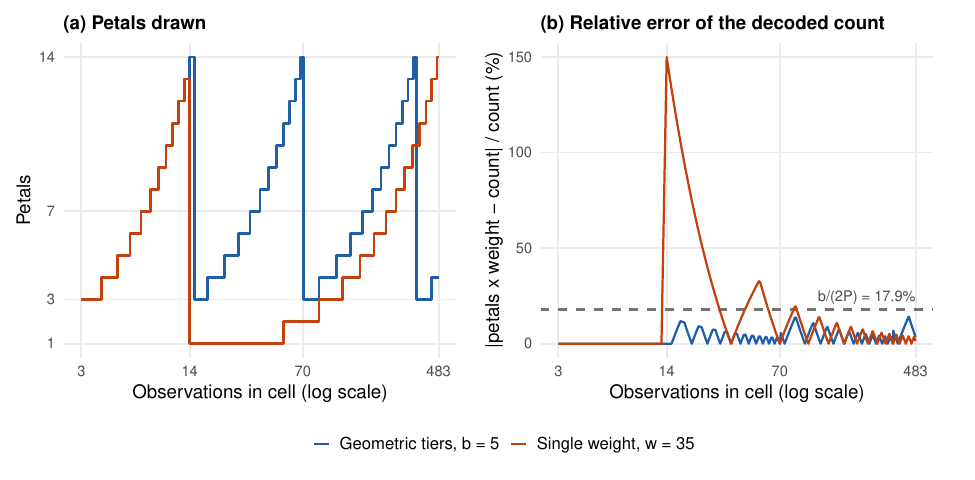}
  \caption{Petals drawn (a) and relative error of the decoded count (b) for every cell count from 3 to 483, the range in \code{psych::bfi}. The single weight $w = 35$ draws one petal for all counts from 14 to 52 and overstates a count of 14 by 150\%; geometric tiers keep 3 to 14 petals and stay below $b/(2P)$.}
  \label{fig:petals}
\end{figure}

\paragraph{Rounding diagnostic.} Rounding loses information. The total absolute rounding error, $\sum |\hat k - k|$ over flower cells, measures how far the glyphs approximate the data. In \code{bfi} it is 91 of 2,757 respondents (3.3\%), against 259 (9.4\%) for the single weight.

\section{Tie diagnostics and automatic selection}\label{sec:select}

Which display is right depends on properties of the data that the scatter plot itself hides, so we measure them first. All diagnostics come from the frequency table of distinct $(x, y)$ pairs. Let $n$ be the number of complete observations, $D$ the number of distinct locations and $k_1, \ldots, k_D$ their counts.

\begin{itemize}
  \item \emph{Hidden fraction} $h = 1 - D/n$, the share of observations drawn on top of another.
  \item \emph{Maximum repeat} $\max_j k_j$, the largest number of observations at one location.
  \item \emph{Distinct locations} $D$, which decides whether every location can carry its own glyph. Above $D = 400$ the glyphs become too small and the data are binned.
  \item \emph{Discreteness}: both variables take at most 30 distinct values. Jitter is acceptable for discrete data and distorting for continuous data.
  \item \emph{Densest cell} $m$, the largest cell count once the cells are chosen (equal to the maximum repeat in exact mode), and the number of tiers it needs, $\lceil \log_b(m/P) \rceil + 1$.
\end{itemize}

The diagnostics cost one \code{count(x, y)}. For data held in a database or in Spark the aggregation runs where the data live, and only the frequency table is collected.

Four ordered checks route a data set to its display (Figure~\ref{fig:flow}). The order reflects cost to the reader: an ordinary scatter plot whenever it hides almost nothing, a small displacement when ties are mild, and petals only when counts matter. Because geometric tiers reduce to the classic and the Dupont--Plummer forms when $m$ is small (Proposition~\ref{prop:bounds}d), one petal rule covers every heavily tied case. The thresholds (5\% hidden, 2 and 5 repeats, 30\% hidden, 400 cells, 30 levels) are defaults, not estimates; Section~\ref{sec:discussion} discusses calibrating them. Every choice is returned with a one-line reason, such as ``densest cell has 483 obs: 4 tiers of 1, 5, \ldots\ obs per petal'', so the analyst sees why a display was chosen and can override it.

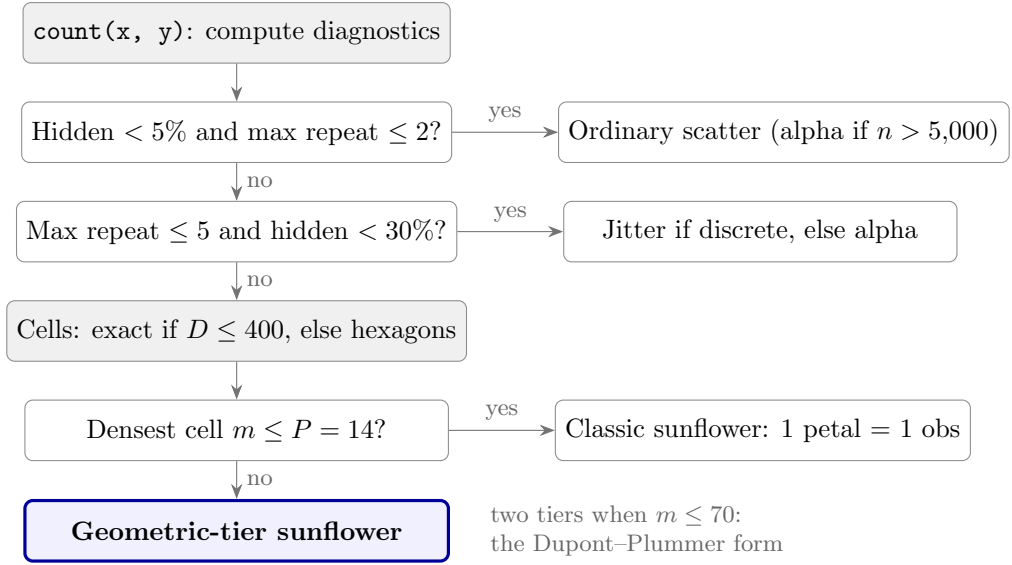
\begin{figure}[!htbp]
  \centering
  \begin{tikzpicture}[
      font=\small, node distance=5mm and 14mm,
      box/.style={draw=black!45, rounded corners=3pt, minimum width=56mm, minimum height=8mm, align=center},
      stepbox/.style={box, fill=black!6},
      outbox/.style={box, minimum width=52mm},
      final/.style={box, draw=blue!60!black, very thick, fill=blue!6, font=\small\bfseries},
      arr/.style={-{Stealth[length=2mm]}, black!55},
      lab/.style={font=\footnotesize, text=black!60}]
    \node[stepbox] (start) {\code{count(x, y)}: compute diagnostics};
    \node[box, below=of start] (q1) {Hidden $< 5\%$ and max repeat $\le 2$?};
    \node[outbox, right=of q1] (o1) {Ordinary scatter (alpha if $n > 5{,}000$)};
    \node[box, below=of q1] (q2) {Max repeat $\le 5$ and hidden $< 30\%$?};
    \node[outbox, right=of q2] (o2) {Jitter if discrete, else alpha};
    \node[stepbox, below=of q2] (cells) {Cells: exact if $D \le 400$, else hexagons};
    \node[box, below=of cells] (q3) {Densest cell $m \le P = 14$?};
    \node[outbox, right=of q3] (o3) {Classic sunflower: 1 petal = 1 obs};
    \node[final, below=of q3] (geo) {Geometric-tier sunflower};
    \node[lab, right=4mm of geo, align=left] {two tiers when $m \le 70$:\\the Dupont--Plummer form};
    \draw[arr] (start) -- (q1);
    \draw[arr] (q1) -- node[lab, right] {no} (q2);
    \draw[arr] (q2) -- node[lab, right] {no} (cells);
    \draw[arr] (cells) -- (q3);
    \draw[arr] (q3) -- node[lab, right] {no} (geo);
    \draw[arr] (q1) -- node[lab, above] {yes} (o1);
    \draw[arr] (q2) -- node[lab, above] {yes} (o2);
    \draw[arr] (q3) -- node[lab, above] {yes} (o3);
  \end{tikzpicture}
  \caption{Decision rules of \code{recommend\_display()} and \code{smart\_scatter()}: four ordered checks route each data set to its display.}
  \label{fig:flow}
\end{figure}

\section{The ggsunflower package}\label{sec:package}

\pkg{ggsunflower} exposes the method as ordinary \pkg{ggplot2} layers, so sunflower plots combine with facets, scales, themes and other layers in the usual way. Table~\ref{tab:api} lists the user-facing functions.

\begin{table}[!htbp]
  \centering\small
  \caption{User-facing functions of \pkg{ggsunflower} 0.1.0.}
  \label{tab:api}
  \begin{tabularx}{\textwidth}{@{}l >{\raggedright\arraybackslash}X@{}}
    \toprule
    Function & Purpose \\
    \midrule
    \code{geom\_sunflower()}, \code{stat\_sunflower()} & sunflower layer; exact cells or hexagonal bins; \code{weight} aesthetic for counts \\
    \code{scale\_fill\_sunflower()} & sequential fill for petal weights, legend ordered by weight \\
    \code{sunflower\_petals()} & the petal rules (\code{"geometric"}, \code{"single"}, \code{"classic"}) applied to a vector of counts \\
    \code{tie\_diagnose()}, \code{recommend\_display()} & tie diagnostics and the display they recommend \\
    \code{smart\_scatter()} & diagnose, choose, draw and report the reason \\
    \code{sunflower\_error()} & rounding error of the sunflowers in a plot \\
    \bottomrule
  \end{tabularx}
\end{table}

\paragraph{Layers.} \code{stat\_sunflower()} assigns observations to cells, applies the petal rule and returns, for each cell, its count, tier, petal weight, number of petals and decoded count; cells below $L$ return their raw points. \code{geom\_sunflower()} draws the hexagons, petals and raw points with \pkg{grid}. The main arguments are \code{cells} (\code{"auto"}, \code{"exact"} or \code{"hex"}), \code{bins}, \code{rule}, \code{base}, \code{max\_petals} ($P$), \code{min\_flower} ($L$) and, for the single-weight rule, \code{light\_max} ($U$) and \code{petal\_weight} ($w$). In exact mode each cell is centred on its own location and the glyph size is set by the most common spacing between distinct values, so open-ended classes coded at half spacing (as in Galton's data) do not shrink every glyph.
\begin{Rcode}
library(ggplot2)
library(ggsunflower)
bfi <- na.omit(psych::bfi[, c("A1", "A2")])
ggplot(bfi, aes(A1, A2)) + geom_sunflower()                  # Figure 1b
ggplot(bfi, aes(A1, A2)) + geom_sunflower(rule = "single")   # Figure 2c
\end{Rcode}

\paragraph{Petal rules.} \code{sunflower\_petals()} applies a rule to a vector of counts, which makes the rules easy to inspect and to use in simulation:
\begin{Rcode}
sunflower_petals(c(14, 52, 483), rule = "single")[, c("count", "weight", "petals")]
#>   count weight petals
#> 1    14     35      1
#> 2    52     35      1
#> 3   483     35     14
sunflower_petals(c(14, 52, 483))[, c("count", "weight", "petals")]
#>   count weight petals
#> 1    14      1     14
#> 2    52      5     10
#> 3   483    125      4
\end{Rcode}

\paragraph{Diagnostics and selection.} \code{tie\_diagnose()} returns the quantities of Section~\ref{sec:select} with the recommended display, and \code{smart\_scatter()} draws it:
\begin{Rcode}
tie_diagnose(bfi$A1, bfi$A2)
#> <tie_diagnosis>
#>   observations   2,757
#>   distinct (D)   36
#>   hidden         98.7%
#>   max repeat     483
#>   discrete       yes
#>   cells          exact, densest has 483 obs (4 tiers)
#>   recommended    geometric
#>                  densest cell has 483 obs: 4 tiers of 1, 5, ... obs per petal
p <- smart_scatter(bfi, A1, A2)
#> geometric-tier sunflower plot: densest cell has 483 obs: 4 tiers of 1, 5, ...
sunflower_error(p)[, c("observations", "abs_error", "rel_error", "max_cell_error")]
#>   observations abs_error  rel_error max_cell_error
#> 1         2757        91 0.03300689      0.1176471
\end{Rcode}

\paragraph{Pre-aggregated and large data.} \code{stat\_sunflower()} accepts a \code{weight} aesthetic, so a frequency table can be plotted directly and gives the same plot as the raw data. Data in a database (\pkg{dbplyr}), in Spark (\pkg{sparklyr}) or in Arrow files can be counted where they live and only the table collected:
\begin{Rcode}
counts <- dplyr::count(dplyr::tbl(con, "survey"), A1, A2) |> dplyr::collect()
ggplot(counts, aes(A1, A2, weight = n)) + geom_sunflower()
\end{Rcode}

\paragraph{Related software.} Base \textsf{R}'s \code{sunflowerplot()} draws one petal per observation. Stata's \code{sunflower} command implements the density-distribution rule \citep{dupont2005}. \code{geom\_hex()} and the \pkg{hexbin} package \citep{carr1987} bin without petals. The \pkg{vayr} package \citep{vayr} provides \code{position\_sunflower()}, which arranges overplotted points in a sunflower-seed spiral; it moves points rather than counting them, and so complements this work.

\section{Examples}\label{sec:examples}

On eight data sets the rules choose three different displays, and each choice is the one a careful analyst would make by hand (Table~\ref{tab:diag}).

\begin{table}[!htbp]
  \centering\footnotesize\setlength{\tabcolsep}{4pt}
  \caption{Diagnostics and the display chosen by \code{recommend\_display()} with default settings. Cells: exact locations or hexagons (15 bins).}
  \label{tab:diag}
  \begin{tabular}{@{}l r r r r l r r l@{}}
    \toprule
    Data (x vs y) & $n$ & $D$ & Hidden & Max rep. & Cells & $m$ & Tiers & Display \\
    \midrule
    \code{faithful}: waiting, eruptions & 272 & 256 & 5.9\% & 2 & exact & 2 & 1 & alpha \\
    \code{MASS::geyser}: waiting, duration & 299 & 257 & 14.0\% & 7 & exact & 7 & 1 & classic \\
    \code{mpg}: city, highway mpg & 234 & 78 & 66.7\% & 14 & exact & 14 & 1 & classic \\
    \code{galton}: parent, child & 928 & 102 & 89.0\% & 48 & exact & 48 & 2 & geometric \\
    \code{bfi}: A1, A2 & 2,757 & 36 & 98.7\% & 483 & exact & 483 & 4 & geometric \\
    Simulated: rounded normal & 50,000 & 167 & 99.7\% & 2,596 & exact & 2,596 & 5 & geometric \\
    \code{diamonds}: carat, price & 53,940 & 28,988 & 46.3\% & 121 & hex & 10,664 & 6 & geometric \\
    \code{diamonds}: carat, depth & 53,940 & 8,716 & 83.8\% & 131 & hex & 14,779 & 6 & geometric \\
    \bottomrule
  \end{tabular}
\end{table}

\paragraph{Negative control.} The Old Faithful waiting and eruption times (\code{faithful}; \citealp{azzalini1990}) have almost no ties, at most two observations at a point, so the rules keep a scatter plot and add only transparency. A tie-aware method should do no harm where it is not needed.

\paragraph{Small discrete data.} No combination of city and highway fuel economy in \code{mpg} \citep{wickham2016} is shared by more than 14 of the 234 cars, so every flower is in tier 0 and the plot is a classic sunflower plot: 33 flowers in which each petal is one car, and 45 sparse locations shown as points (Figure~\ref{fig:mpg}). Here \code{geom\_count()} conveys the same pattern, but only the petals give exact counts.

\begin{figure}[!htbp]
  \centering
  \includegraphics[width=\textwidth]{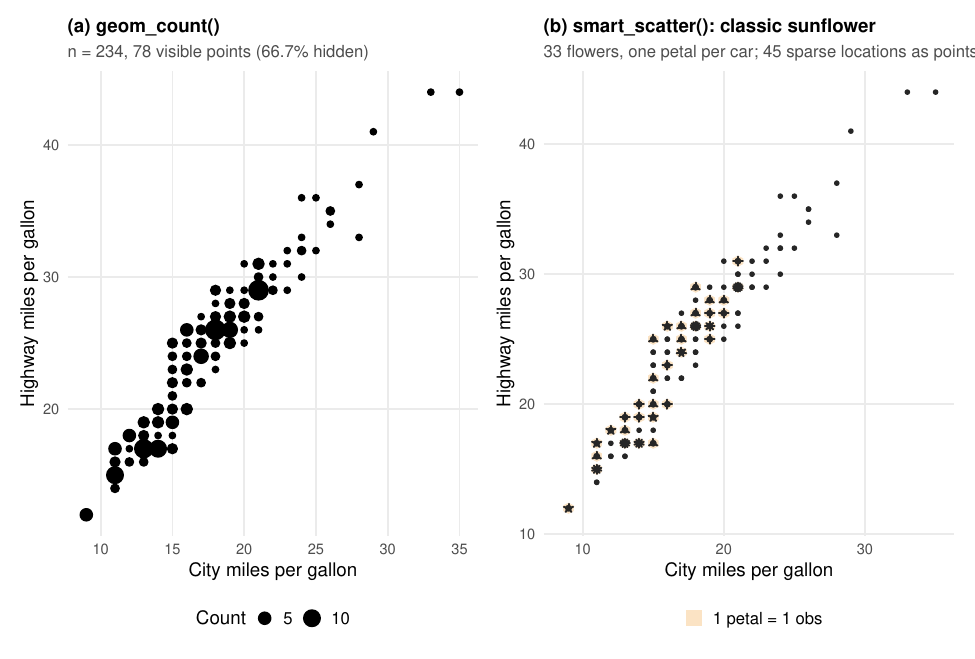}
  \caption{\code{ggplot2::mpg}. (a) \code{geom\_count()}. (b) The classic sunflower plot chosen by \code{smart\_scatter()}: one petal per car.}
  \label{fig:mpg}
\end{figure}

\paragraph{Survey data.} In \code{bfi}, with counts from 3 to 483 at 36 locations, four tiers are needed (Figures~\ref{fig:bfi} and~\ref{fig:bfi4}). The geometric-tier plot shows at once that respondents concentrate where A1 is low and A2 is high, as expected because A1 is reverse-keyed, and each count can be read off: four petals of 125 at $(1, 6)$, close to the true 483.

\paragraph{Historical data.} Galton's 928 parent--child pairs (\code{galton} in \pkg{psychTools}; \citealp{psychtools}), recorded in one-inch classes, fall on 102 locations; 89\% of the children are hidden by the ordinary scatter plot (Figure~\ref{fig:galton}). The densest cell holds 48 children, so two tiers suffice, and the 30 sparsest locations (41 children) keep their raw points. The plot shows what Galton's own table of counts showed and the ordinary scatter plot hides: the elliptical concentration of the data around the regression line.

\begin{figure}[!htbp]
  \centering
  \includegraphics[width=\textwidth]{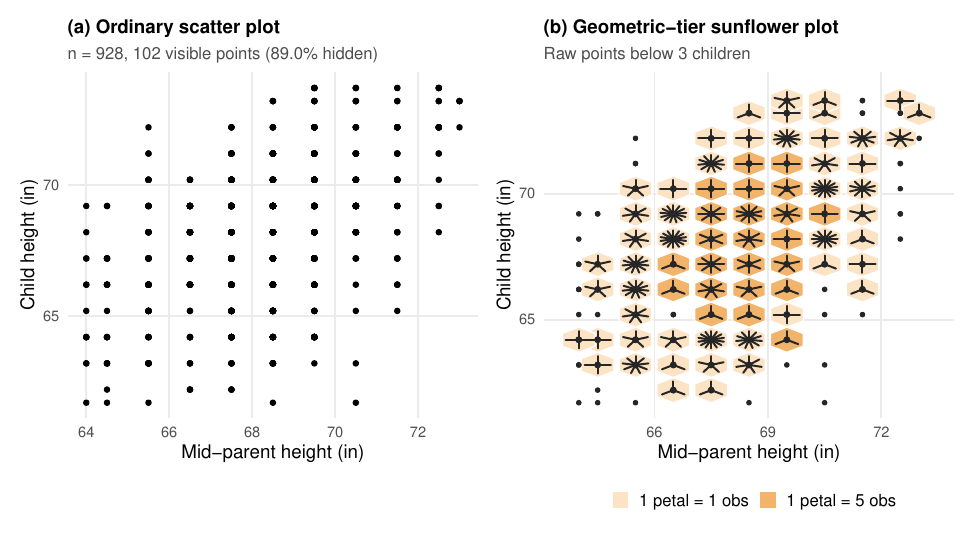}
  \caption{Galton's parent--child heights. (a) Ordinary scatter plot: 102 visible points for 928 children. (b) Geometric-tier sunflower plot with two tiers.}
  \label{fig:galton}
\end{figure}

\paragraph{Large data.} Carat and depth of 53,664 diamonds (\code{diamonds} in \pkg{ggplot2}, restricted to carat $\le 2.5$ and depth 56--68\%) fall on 8,457 distinct locations, too many for a glyph each, so the data are binned into hexagons (Figure~\ref{fig:diamonds}). The densest hexagon holds 4,993 diamonds and five tiers are used, from 1 to 625 diamonds per petal, while 40 diamonds in sparse hexagons keep their raw points. Hexagonal binning with colour (panel~b) gives the quickest impression of where the mass lies; the sunflower plot adds readable counts and keeps the sparse tail as points instead of absorbing it into faint bins. We regard the two as complementary for large continuous data.

\begin{figure}[!htbp]
  \centering
  \includegraphics[width=\textwidth]{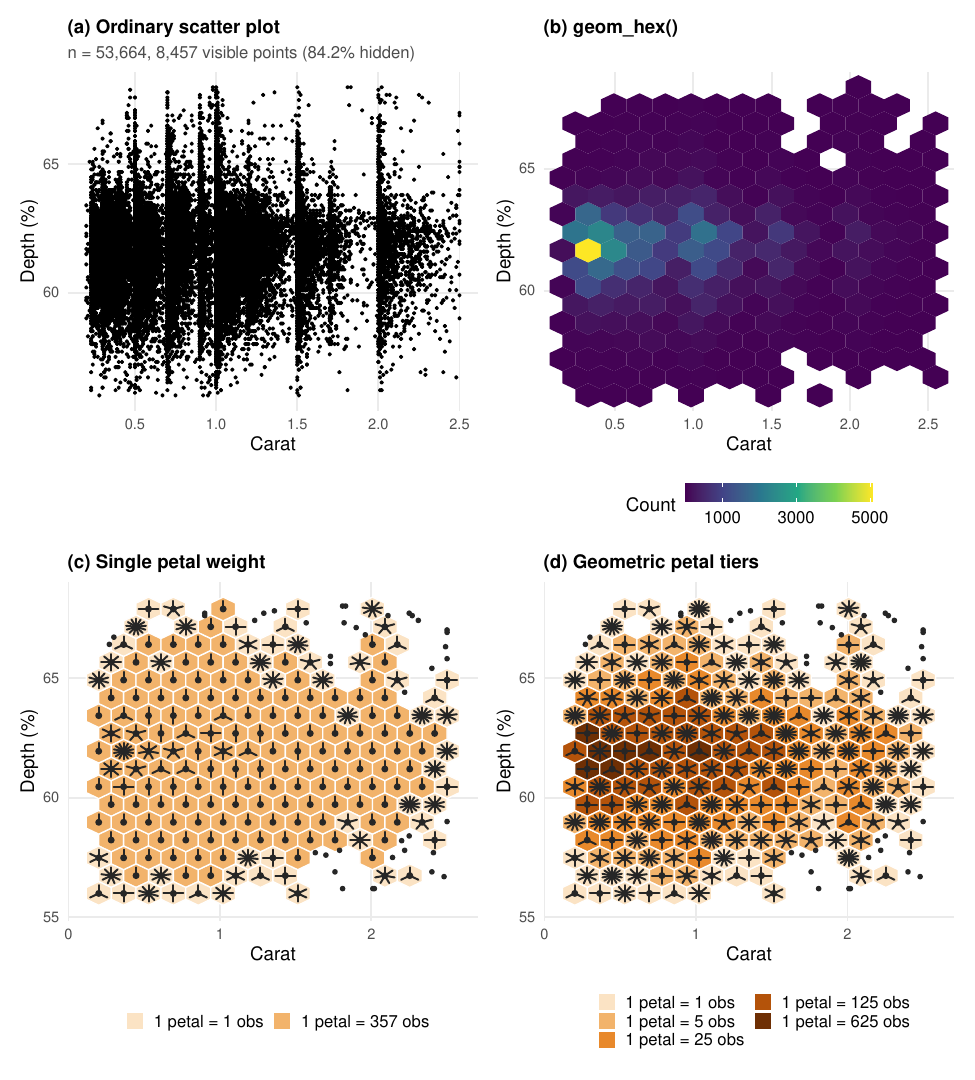}
  \caption{\code{ggplot2::diamonds}, carat against depth. (a) Ordinary scatter plot. (b) \code{geom\_hex()}. (c) Single dark weight: most dense hexagons collapse to one or two petals. (d) Geometric petal tiers.}
  \label{fig:diamonds}
\end{figure}

Across the three examples drawn with geometric tiers the total rounding error is between 3\% and 4\% of the observations: 91 of 2,757 in \code{bfi}, 29 of 928 for Galton and 2,082 of 53,664 diamonds.

\section{Simulation study}\label{sec:sim}

Across 35 simulated conditions, geometric tiers kept every decoded count within 14\% of the truth, while the single-weight rule's worst cell was off by a median of 251\%, and its error grew with the size of the data.

\paragraph{Design.} Pairs $(x, y)$ were drawn from a bivariate normal distribution or a bivariate $t$ distribution with three degrees of freedom (heavier tails, so more skewed cell counts), with correlation $\rho \in \{0, 0.6\}$. Each variable was discretised into $G \in \{6, 12, 24\}$ equal-width classes over its central 99.5\%, giving a $G \times G$ grid of exact cells, and the sample size was $n \in \{10^3, 10^4, 10^5\}$. Each of the 36 conditions was replicated 100 times. In one condition ($n = 1{,}000$, $G = 24$, $\rho = 0$, normal) the densest cell averaged fewer than 14 observations, so the two rules rarely differed; it is omitted, leaving 35. Both rules used $L = 3$ and $P = 14$; the single-weight rule used $U = 13$ and $w = \lceil m/14 \rceil$, and geometric tiers used $b = 5$.

\paragraph{Measures.} For each plot, over the cells drawn as flowers: the \emph{worst-cell relative error} $\max |\hat k - k|/k$; the share of flowers \emph{misread}, decoded more than 50\% away from their count; the share \emph{collapsed} to one or two petals; and \emph{resolution}, the number of distinct decoded counts divided by the number of distinct true counts (1 means that every different count looks different).

\begin{table}[!htbp]
  \centering\small
  \caption{Simulation results by sample size: means over conditions (35 in all) of per-plot measures averaged over 100 replicates. SW = single weight (Dupont--Plummer form); GT = geometric tiers.}
  \label{tab:sim}
  \begin{tabular}{@{}r r r rr rr rr rr@{}}
    \toprule
    & & Mean & \multicolumn{2}{c}{Worst-cell error} & \multicolumn{2}{c}{Misread} & \multicolumn{2}{c}{Collapsed} & \multicolumn{2}{c}{Resolution} \\
    \cmidrule(lr){4-5}\cmidrule(lr){6-7}\cmidrule(lr){8-9}\cmidrule(l){10-11}
    $n$ & Cond. & densest & SW & GT & SW & GT & SW & GT & SW & GT \\
    \midrule
    1,000   & 11 & 147    & 26\%    & 10\% & 0.8\% & 0\% & 11\% & 0\% & 0.75 & 0.81 \\
    10,000  & 12 & 1,240  & 501\%   & 13\% & 19\%  & 0\% & 38\% & 0\% & 0.33 & 0.54 \\
    100,000 & 12 & 11,237 & 3,767\% & 13\% & 45\%  & 0\% & 65\% & 0\% & 0.14 & 0.37 \\
    \bottomrule
  \end{tabular}
\end{table}

Under geometric tiers the mean worst-cell error was at most 14.0\% in every condition, and by Proposition~\ref{prop:bounds} no plot can exceed 17.9\%; no flower was misread or collapsed (Table~\ref{tab:sim}). The single-weight rule misread 22\% of flowers and collapsed 39\% on average, and both shares grow with $n$ because the weight grows with the densest cell while the smallest dark cell stays at 14 observations (Figure~\ref{fig:sim}). Its worst condition, bivariate $t$ with $n = 100{,}000$, $\rho = 0.6$ and $G = 6$, overstated a cell by 14,874\% on average. Both rules lose resolution as $n$ grows, because a flower can show only about a dozen values per tier, but geometric tiers keep 1.1 to 2.6 times as many counts distinguishable. For small $m$ the two rules coincide or nearly so (Proposition~\ref{prop:bounds}d), which is why their points overlap at the left of Figure~\ref{fig:sim}.

\begin{figure}[!htbp]
  \centering
  \includegraphics[width=0.85\textwidth]{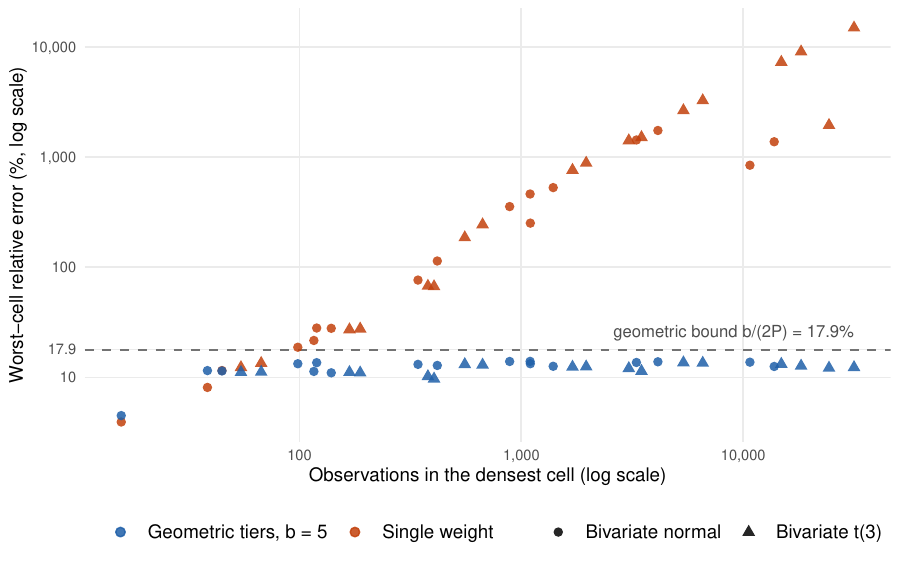}
  \caption{Worst-cell relative error against the densest cell, one point per condition (mean over 100 replicates). The single-weight error grows with $m$; geometric tiers stay below $b/(2P) = 17.9\%$.}
  \label{fig:sim}
\end{figure}

\section{Discussion}\label{sec:discussion}

A scatter plot of tied data can hide almost all of it, and the analyst usually cannot tell from the plot. The diagnostics make that visible in one number, and geometric tiers give a display whose counts can be read to within a known error whatever the size of the data.

\paragraph{Where the method helps most.} The gain is largest for discrete data with heavy repetition, such as survey items, counts, coded categories and rounded measurements, where every location of a bounded grid can be occupied and the ordinary scatter plot degenerates into a uniform lattice. For large continuous data the sunflower plot complements rather than replaces hexagonal binning: binning gives the quicker overall impression, and petals give readable counts and keep sparse observations visible.

\paragraph{Limitations.} Glyphs need room: with many cells, flowers become small; \code{mpg}, with 78 locations spread over a $27 \times 33$ range of integer values, already needs a large figure, and binned mode limits the number of cells only by coarsening the data. A reader must tell the tiers apart by fill; four or five shades of one sequential hue are distinguishable, which with $b = 5$ covers cells up to 8,750 observations, and beyond that a larger base is needed at the cost of a looser bound. At a tier boundary, a cell just above $P b^{t}$ is drawn with about $P/b$ petals in a darker shade next to a cell just below it with $P$ petals in a lighter one; the order is carried by the shade, but whether readers decode it correctly is an empirical question. Finally, the selection thresholds are reasoned defaults rather than estimates.

\paragraph{Future work.} A perception study in the tradition of \citet{cleveland1984perception} would test whether readers estimate and compare counts more accurately from geometric-tier sunflowers than from point area, colour-binned hexagons or the single-weight rule, and would calibrate the selection thresholds. Ties can also mislead about association: in \code{bfi} the 36 visible dots form a complete grid with no apparent trend, although the 2,757 respondents correlate at $r = -0.34$. Whether tie-aware displays improve judgements of association, not only of density, is a natural extension.

\paragraph{Conclusion.} One petal weight cannot serve cells of 14 and of 483 observations. Letting each cell choose its weight from a geometric ladder bounds the error of every decoded count, reduces to the classic and Dupont--Plummer forms where those suffice, and, with simple tie diagnostics to decide when it is needed, turns the scatter plot of heavily tied data from a picture of where the data \emph{could} be into a picture of where they \emph{are}.

\section*{Software and reproducibility}
\pkg{ggsunflower} 0.1.0 is available at \url{https://github.com/sameeranc/ggsunflower} under the MIT licence. All figures and tables can be reproduced with the scripts \code{figures.R}, \code{simulation.R} and \code{fig8\_table3.R} in the supplementary material, using \textsf{R} 4.3.3, \pkg{ggplot2} 3.4.4, \pkg{psych} 2.4.1 and \pkg{psychTools} 2.4.2.

\section*{Use of generative AI}
The author used Claude, a large language model developed by Anthropic, as a research and programming assistant. It was used in three ways: in discussions that shaped the method, including the geometric-tier rule and its error bound; to write the code for the \pkg{ggsunflower} package, the simulation study and the figures; and to draft and edit the text of this paper. The author set the research question, made the decisions on the method and its presentation, checked all code, results and references, and is responsible for the paper and the package.

\bibliographystyle{plainnat}
\bibliography{references}

@article{carr1987,
  author  = {Carr, D. B. and Littlefield, R. J. and Nicholson, W. L. and Littlefield, J. S.},
  title   = {Scatterplot Matrix Techniques for Large {N}},
  journal = {Journal of the American Statistical Association},
  year    = {1987},
  volume  = {82},
  number  = {398},
  pages   = {424--436},
  doi     = {10.1080/01621459.1987.10478445}
}

@article{cleveland1984faces,
  author  = {Cleveland, William S. and McGill, Robert},
  title   = {The Many Faces of a Scatterplot},
  journal = {Journal of the American Statistical Association},
  year    = {1984},
  volume  = {79},
  number  = {388},
  pages   = {807--822},
  doi     = {10.1080/01621459.1984.10477098}
}

@article{cleveland1984perception,
  author  = {Cleveland, William S. and McGill, Robert},
  title   = {Graphical Perception: Theory, Experimentation, and Application to the Development of Graphical Methods},
  journal = {Journal of the American Statistical Association},
  year    = {1984},
  volume  = {79},
  number  = {387},
  pages   = {531--554},
  doi     = {10.1080/01621459.1984.10478080}
}

@article{dupont2003,
  author  = {Dupont, William D. and Plummer, Jr., W. Dale},
  title   = {Density Distribution Sunflower Plots},
  journal = {Journal of Statistical Software},
  year    = {2003},
  volume  = {8},
  number  = {3},
  doi     = {10.18637/jss.v008.i03}
}

@article{dupont2005,
  author  = {Dupont, William D. and Plummer, Jr., W. Dale},
  title   = {Using Density-Distribution Sunflower Plots to Explore Bivariate Relationships in Dense Data},
  journal = {The Stata Journal},
  year    = {2005},
  volume  = {5},
  number  = {3},
  pages   = {371--384},
  doi     = {10.1177/1536867X0500500305}
}

@article{galton1886,
  author  = {Galton, Francis},
  title   = {Regression Towards Mediocrity in Hereditary Stature},
  journal = {Journal of the Anthropological Institute of Great Britain and Ireland},
  year    = {1886},
  volume  = {15},
  pages   = {246--263},
  doi     = {10.2307/2841583}
}

@book{wickham2016,
  author    = {Wickham, Hadley},
  title     = {ggplot2: Elegant Graphics for Data Analysis},
  edition   = {2},
  publisher = {Springer-Verlag},
  address   = {New York},
  year      = {2016},
  doi       = {10.1007/978-3-319-24277-4}
}

@manual{psych,
  author = {Revelle, William},
  title  = {psych: Procedures for Psychological, Psychometric, and Personality Research},
  note   = {R package version 2.4.1},
  organization = {Northwestern University},
  address = {Evanston, Illinois},
  year   = {2024},
  url    = {https://CRAN.R-project.org/package=psych}
}

@manual{psychtools,
  author = {Revelle, William},
  title  = {psychTools: Tools to Accompany the 'psych' Package for Psychological Research},
  note   = {R package version 2.4.2},
  organization = {Northwestern University},
  address = {Evanston, Illinois},
  year   = {2024},
  url    = {https://CRAN.R-project.org/package=psychTools}
}

@article{azzalini1990,
  author  = {Azzalini, A. and Bowman, A. W.},
  title   = {A Look at Some Data on the {Old Faithful} Geyser},
  journal = {Journal of the Royal Statistical Society. Series C (Applied Statistics)},
  year    = {1990},
  volume  = {39},
  number  = {3},
  pages   = {357--365},
  doi     = {10.2307/2347385}
}

@manual{vayr,
  author = {Coppock, Alexander},
  title  = {vayr: Visualize as You Randomize},
  note   = {R package version 1.1.0},
  year   = {2026},
  url    = {https://acoppock.r-universe.dev/vayr}
}

@manual{rcore,
  author       = {{R Core Team}},
  title        = {R: A Language and Environment for Statistical Computing},
  organization = {R Foundation for Statistical Computing},
  address      = {Vienna, Austria},
  year         = {2024},
  url          = {https://www.R-project.org/}
}

\end{document}